\documentclass[12pt]{article}

\usepackage{amsmath,amssymb,amsthm}
\usepackage{geometry}
\usepackage{graphicx}
\usepackage{epstopdf}
\usepackage{booktabs}
\usepackage{algorithm}
\usepackage{algpseudocode}
\usepackage[dvipsnames]{xcolor}
\usepackage{hyperref}
\usepackage[numbers,sort&compress]{natbib}

\hypersetup{colorlinks=true, linkcolor=blue!70!black,
            urlcolor=blue!70!black, citecolor=blue!70!black}

\theoremstyle{plain}
\newtheorem{theorem}{Theorem}[section]
\newtheorem{lemma}[theorem]{Lemma}

\theoremstyle{definition}
\newtheorem{definition}[theorem]{Definition}

\theoremstyle{remark}

\title{Apple Peel Unfolding of Archimedean and Catalan Solids}

\author{
  T.~Yoshino\thanks{Corresponding author.
    Email: \texttt{tyoshino@toyo.jp}}\\
  \small Department of Mechanical Engineering, Toyo University,\\
  \small 2100 Kujirai, Kawagoe, 350-8585, Japan
  \and
  S.~Chaidee\\
  \small Department of Mathematics, Faculty of Science,
         Chiang Mai University,\\
  \small 239 Huay Kaew Road, Muang District,
         Chiang Mai, 50200, Thailand\\
  \small Advanced Research Center for Computational Simulation,
         Chiang Mai University,\\
  \small 239 Huay Kaew Road, Muang District,
         Chiang Mai, 50200, Thailand
}

\date{2026/04/16}

\begin{document}
\maketitle


\bigskip

\begin{abstract}
We consider a new treatment for making polyhedron nets referred to as
``apple peel unfolding'': drawing the nets as if we were peeling off apple
skins.
We define apple peel unfolding strictly and implement a program that
derives the sequential selection of the polyhedral faces for a target
polyhedron in accordance with the definition.
Consequently, the program determines whether the polyhedron is peelable
(can be peeled completely).
We classify Archimedean solids and their duals (Catalan solids) as
perfect (always peelable), possible (peelable for restricted cases),
or impossible.
The results show that three Archimedean and six Catalan solids are
perfect, and three Archimedean and three Catalan ones are possible.
\end{abstract}

\tableofcontents

\section{Introduction}\label{introduction}

Peeling (or unwrapping) an object, such as a fruit or package, is a
quite common procedure in our daily lives.
Such peeling leads to plane figures in most cases.
Unfolding polyhedral surfaces is a well-known subject as a method to
obtain polyhedron nets, and has been widely studied from a mathematical
point of view, as compiled in~\cite{Uehara2020}.
It seems reasonable to consider the peeling of an apple or orange as
similar to the unfolding problem of a three-dimensional object.
That is, given the axis of an object, we peel off a surface in such a
way that the surface is flattened or almost flattened to a plane.

Mathematical consideration of the peeling of an object has been described
in various studies.
\cite{Bartholdi2012} considered peeling off the unit sphere and analyzed
the corresponding spiral curve.
In the case of polyhedra, a simplified structure of surfaces,
one consideration similar to peeling apples is zipper unfolding---slicing
a polyhedron along a single cut path like a zipper in a piece of
clothing---proposed by \cite{Demaine2010}, in which a polyhedron net is
acquired based on a Hamiltonian path of the polyhedron.
Spiral unfolding~\cite{ORourke2015} expands the surfaces of polyhedra by
cutting their faces spirally.

Apple peeling, the unfolding of a polyhedron such that the polyhedron net
is spiral, is an interesting unfolding method of polyhedra.
Apple peeling of Platonic solids has already been examined and that
consideration was extended to the higher-dimensional polytopes
\cite{Kaino2019}.
This procedure can be regarded as sequentially and spirally unfolding the
outermost layer of facets with respect to an axis.
Furthermore, Kaino demonstrated two results for the dodecahedron, as
discussed later.
However, there has been no systematic treatment of such unfolding, and a
rigorous definition of apple peeling has yet to be clarified.

One application of the polyhedron peeling problem is the study of
fullerenes, including the naming of the polyhedrons representing the
molecular structure.
In particular, the studies of Manolopoulos~\cite{Manolopoulos1991} focused
on the peeling of fullerenes, considering them as cubic polyhedra, together
with examples of polyhedra that cannot be peeled~\cite{Manolopoulos1993}.
These topics were later investigated and the results significantly updated
in the study of \cite{WSA2018} to clarify a general algorithm for finding
a face-spiral of a polyhedron.

Our objectives in the present research are (a) to define the apple peeling
of a polyhedron strictly and obtain results for Archimedean and Catalan
solids and classify those results, and (b) to draw polyhedron nets and
planar graphs representing the peelings.
The former objective is necessary to extend the topic to arbitrary convex
polyhedra and to higher-dimensional polytopes (e.g., the six regular
polytopes in four-dimensional space), although we focus on only
three-dimensional polyhedra in the present study.
It should be noted that apple peel unfolding is related to graph theory
because peeling polyhedra is equivalent to finding Hamiltonian paths of
the dual polyhedra.
However, solutions to these problems cannot be obtained from their planar
graphs directly due to the non-existence of their axis.

The remainder of this article is organized as follows.
The definition and the algorithm of apple peel unfolding are described in
Section~\ref{sec:methods}.
The results of numerical simulations of the unfoldings of Archimedean and
Catalan solids are explained in Section~\ref{sec:results}.
Finally, we discuss the results in Section~\ref{sec:discussions}.

\section{Methods}\label{sec:methods}

\subsection{Preliminaries}

For a considered polyhedron $\mathcal{P}$, we denote the sets of faces,
vertices, and edges as $F=\{f_1, \ldots, f_n\}$,
$V=\{v_1, \ldots, v_m\}$, and $E=\{e_1, \ldots, e_l\}$, respectively.
Note that $n$ denotes the number of faces, and $f_i$ consists of $k_i$
integers if the $i$-th polygon is a $k_i$-gon; $m$ and $l$ denote the
numbers of vertices and edges of the polyhedron, respectively.

It is well-known and fruitful to consider a polyhedron as a graph by
associating the set of vertices, edges, and faces of a polyhedron
$\mathcal{P}$ with the graph $G=(V, E)$.
The degree of a vertex $v$ of $V$ is the number of edges adjacent to $v$.
By Steinitz's theorem, a graph $G$ is a graph of a polyhedron if and only
if $G$ is a simple planar graph that is 3-connected.
The graph $G$ is called a polyhedral graph or skeleton.
Based on the graph information, the Hamiltonian path is a path consisting
of a sequence of vertices in which each vertex is visited exactly once.

A net of a polyhedron $\mathcal{P}$ is a planar figure derived by cutting
the polyhedral surface.
It is called an e-net if the cut is through the edges of a polyhedron.
To avoid confusion, we use the term ``net'' to mean an e-net in the
present article.
Figure~\ref{Fig:Preliminaries} shows examples of two polyhedra,
a truncated icosahedron (Fig.~\ref{Fig:Preliminaries}a) and a pentakis
dodecahedron (Fig.~\ref{Fig:Preliminaries}d), with associated polyhedron
nets (Figs.~\ref{Fig:Preliminaries}b and~\ref{Fig:Preliminaries}e) and
polyhedral graphs (Figs.~\ref{Fig:Preliminaries}c
and~\ref{Fig:Preliminaries}f).

For a given polyhedron $\mathcal{P}$, the dual polyhedron $\mathcal{P}^*$
is the polyhedron for which each face of $\mathcal{P}$ is represented by a
vertex of $\mathcal{P}^*$, and each pair of adjacent faces of $\mathcal{P}$
represents an edge connecting vertices of $\mathcal{P}^*$.
Given any polyhedron $\mathcal{P}$, its dual can be immediately constructed.
The concept of duality is important for two reasons: one is that the
polyhedra we focus on have the duality relation; the other is that the
apple-peel unfolding of a given solid is equivalent to the search for the
Hamiltonian path of its dual solid with some restrictions.

To specify the solid we are interested in, a polyhedron is said to be an
Archimedean solid, or a semi-regular polyhedron, if its faces are regular
polygons and the number of faces around each vertex is the same.
The dual Archimedean solids are called Catalan solids.
We represent the individual Archimedean and Catalan solids by using curly
brackets and square brackets, respectively, as follows:
$\{n_1, n_2, \ldots, n_q\}$ and $[n_1, n_2, \ldots, n_q]$ refer to
Archimedean and Catalan solids, respectively.
The former representation indicates each vertex of the polyhedron as being
the joint of an $n_1$-gon, $n_2$-gon, $\ldots$, and $n_q$-gon and the
latter indicates that each face of the polyhedron has vertices of degrees
of $n_1$, $n_2$, $\ldots$, and $n_q$.
Therefore, two polyhedra with a dual relation are represented by the same
number combination with different brackets.

\begin{figure}[ht]
\centering
\resizebox*{14cm}{!}{\includegraphics{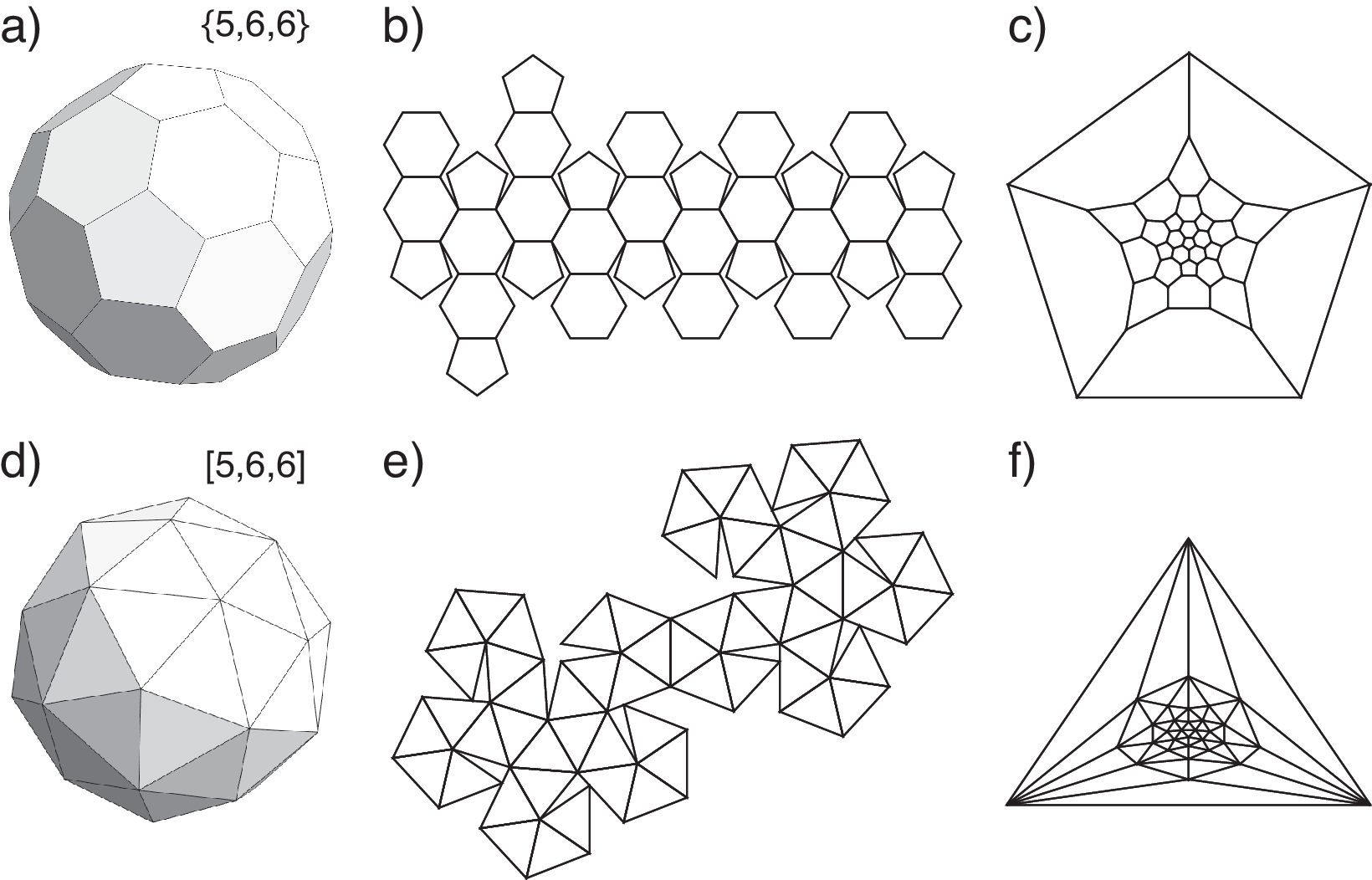}}
\caption{(a) Truncated icosahedron $\{5,6,6\}$,
  (b) net of truncated icosahedron,
  (c) polyhedral graph of truncated icosahedron,
  (d) pentakis dodecahedron $[5,6,6]$ (dual of truncated icosahedron),
  (e) net of pentakis dodecahedron,
  (f) polyhedral graph of pentakis dodecahedron.}
\label{Fig:Preliminaries}
\end{figure}

Polytope unfolding appears widely in discrete and computational geometry
problems.
In the present study, we are interested in the unfolding of a polytope
following the apple peeling procedure presented in~\cite{Bartholdi2012}
for the case of a sphere, and following~\cite{Kaino2019} for the faces of
the Platonic solids and the regular polytopes in $\mathbb{R}^4$.

\subsection{Apple Peeling Procedure}

To follow the apple peeling procedure in the real world, we assume that an
apple whose skin is peeled off by a knife is rotated along a fixed direction
defined by an axis of peeling recognized as the apple stalk.
Mathematically, the axis of polyhedron peeling is determined by a fixed ray
emanating from the center of the polyhedron while the polyhedron is rotated
along a fixed axis direction.
Hence, it is necessary to define two other orthogonal rays, which are also
orthogonal to the given axis.
Remark that the peeling in this context satisfies the unfolding of a
polyhedron to a net so that the edges of the polyhedron connect the faces
of the polyhedron.

Therefore, we unfold polyhedral faces according to a fixed axis containing
the polyhedron's centroid.
The selection of all faces is sequentially chosen clockwise (or
anti-clockwise) around the axis from top to bottom.
The selection's direction corresponds to the handler's dominant hand
hereafter.
We assume that the handler is right-handed.

Assume that a three-dimensional polyhedron $\mathcal{P}$ is located in the
three-dimensional space $\mathbb{R}^3$.
A pair of neighboring faces $F_1$ and $F_2$ is chosen arbitrarily, and the
apple peeling of the facets starts from $F_1$, and then $F_2$ is selected
as the next face to be peeled off.

For convenience, we set the polyhedron $\mathcal{P}$ so that the origin
$\mathrm{O}$ is located at the centroid of $\mathcal{P}$.
The Cartesian coordinate O-$x_1x_2x_3$ is set so that the $x_3$-axis
contains the centroid of $F_1$.
Then, the axis of peeling of the polyhedron is defined as the $x_3$-axis,
which contains the centroids of both $\mathcal{P}$ and $F_1$.
Specifically, the coordinates of the centroid of $F_1$ are defined as
$(0, 0, x_{3,1})$ and we set the direction of the $x_3$-axis according to
$x_{3,1}>0$.
Therefore, we regard that $F_1$ is set on the top.
As an example of the above definition, we consider the case of the
three-dimensional space shown in the schematic illustration of a truncated
icosahedron in Fig.~\ref{Fig:definition}.

After the faces $F_1$ and $F_2$ are selected, the next candidates, the
neighboring faces of the previous face in the sequence, are chosen
according to the following priority.
If the number of candidates is only one, we use that candidate.
If there are multiple candidates, then we select the face whose centroid
is at the highest position among the faces on \textit{the left side}.
If there are no candidates on the left side, we select the polygon whose
centroid is located the lowest among the polygons on the right side.
Finally, we continue the selection process until all the faces have been
selected or there are no unselected neighbors.

\begin{figure}[ht]
\centering
\resizebox*{6cm}{!}{\includegraphics{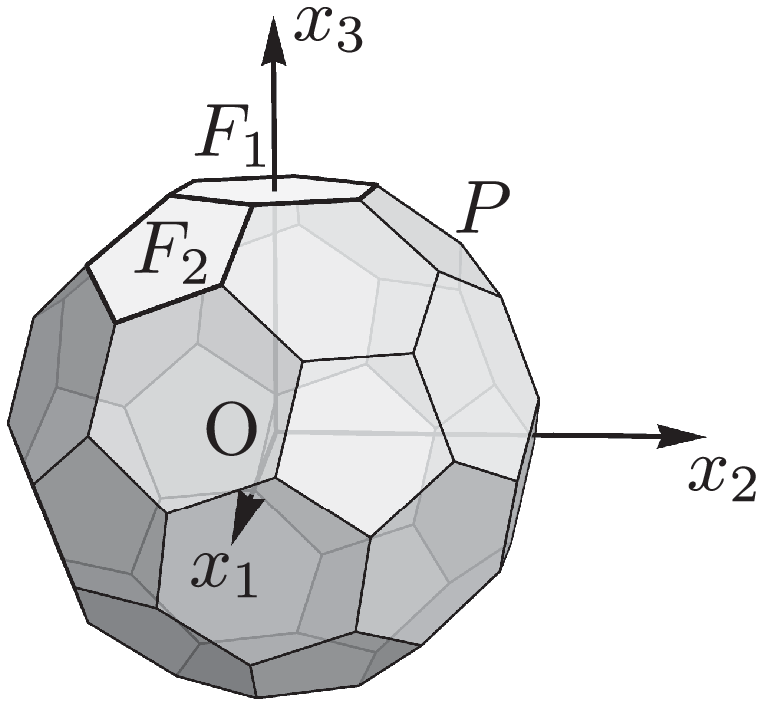}}
\caption{Definition of coordinates in three-dimensional space.
  The $x_3$-axis is the axis of peeling.}
\label{Fig:definition}
\end{figure}

Figure~\ref{Fig:Example3D} shows an example of the intermediate state of
the unfolding procedure of a truncated icosahedron.
The arrow shows the $x_3$-axis of peeling.
The planar connection of gray polygons indicates the intermediate state of
its polyhedron net.
Yellow and green polygons on the polyhedron are the unfolded and not
unfolded faces, respectively.
We represent the last selected face with red.
The image illustrates that the peeling was started from a hexagon, the
handler chose a pentagon as the second polygon, and the peeling continued
until the selection of the eighteenth face.

\begin{figure}[ht]
\centering
\resizebox*{8cm}{!}{\includegraphics{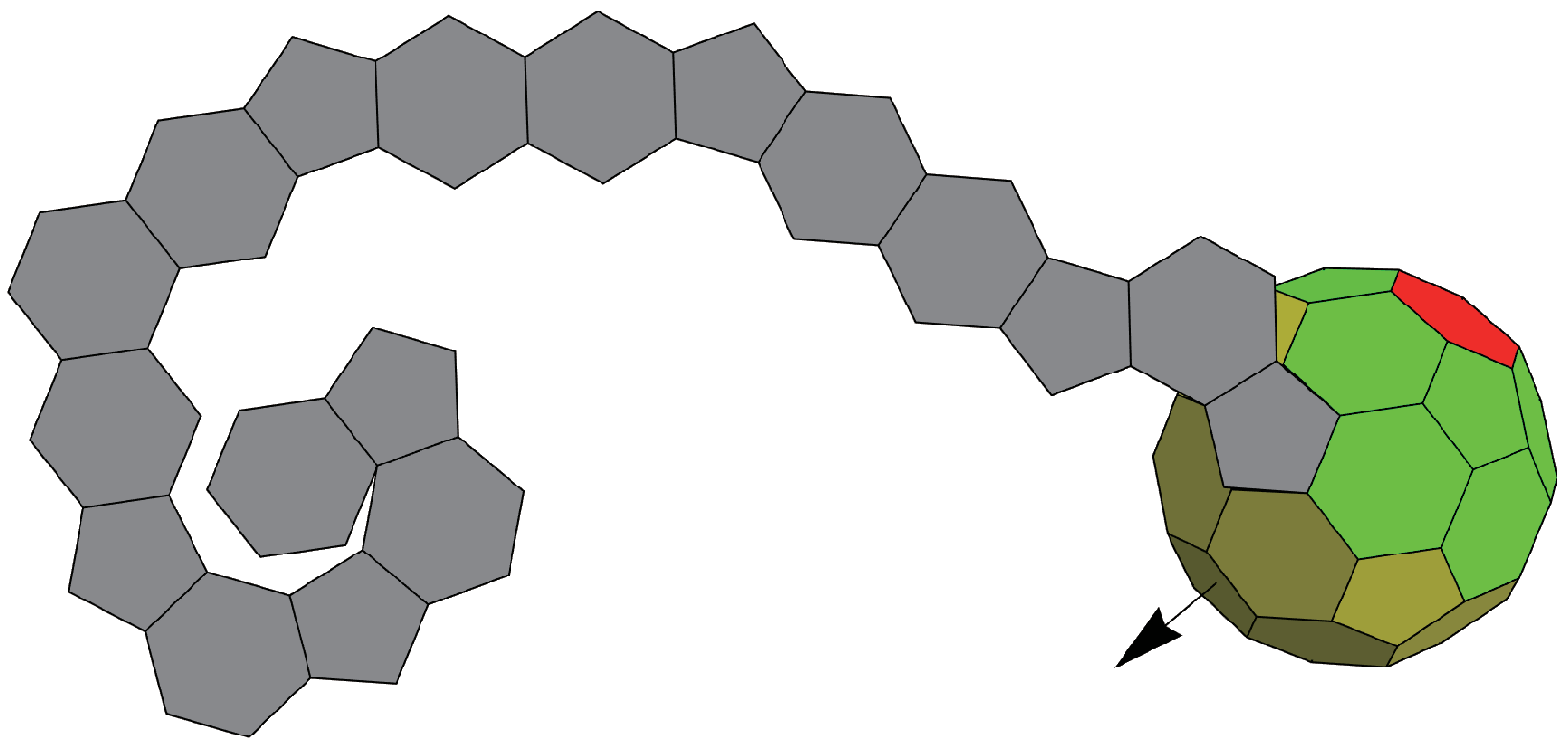}}
\caption{Example of intermediate state of apple peel unfolding of a
  truncated icosahedron. The arrow represents the axis of peeling.}
\label{Fig:Example3D}
\end{figure}

\begin{figure}[ht]
\centering
\resizebox*{8cm}{!}{\includegraphics{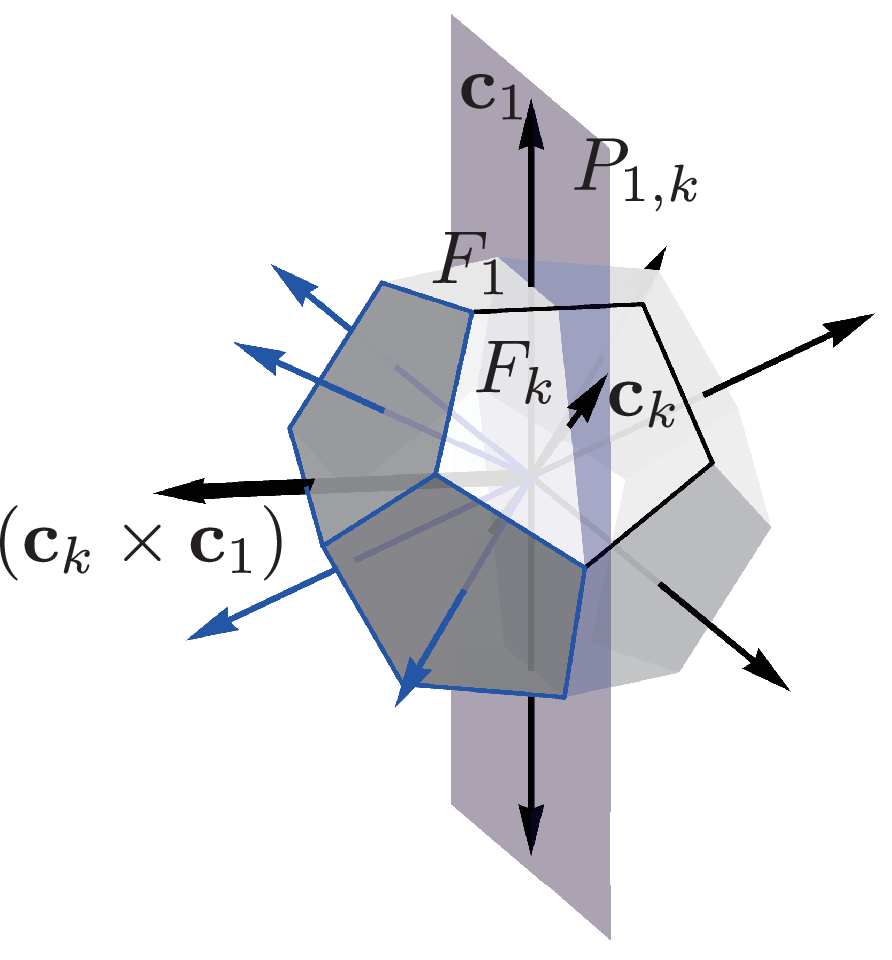}}
\caption{Example of consideration of the left side of a point $P_{1,k}$
  with respect to the plane defined by the two faces $F_1$ and $F_k$.}
\label{Fig:leftside}
\end{figure}

During the peeling procedure, the determination of being on the
\textit{`left side'} with respect to the $x_3$-axis and the considered
face $F_k$ needs to be clarified rigorously, as illustrated in
Fig.~\ref{Fig:leftside}.
Suppose that $\mathbf{c}_1$ is a vector proportional to the vector from
the origin to $c_1$, the centroid of face $F_1$, which is the starting
face.
For the $F_k$ selected during the procedure with centroid $c_k$, we
construct a vector $\mathbf{c}_k$ as a vector proportional to the vector
from the origin to $c_k$.
Then the plane $P_{1,k}$ containing the origin, $c_1$, and $c_k$ can be
defined uniquely and the normal vector of $P_{1,k}$ is obtained as
$(\mathbf{c}_k \times \mathbf{c}_1)$, as shown in the figure.
To consider whether a point $P(x, y, z)$ is on the left-hand side with
respect to the choice of $F_k$, we verify whether
$(\mathbf{c}_k \times \mathbf{c}_1) \cdot \overrightarrow{OP} > 0$,
i.e., the angle between the vector $(\mathbf{c}_k \times \mathbf{c}_1)$
of plane $P_{i,k}$ and vector $\overrightarrow{OP}$ is acute, as shown
by the blue arrows in Fig.~\ref{Fig:leftside}.

\subsection{Algorithm for Apple Peeling}

The algorithm for numerical simulation is presented as
Algorithm~\ref{peeling}.
Input data consist of three sets:
(1) vertex coordinates $V=\{v_1, v_2, \ldots, v_m\}$;
(2) vertex indices of all faces $F=\{f_1, f_2, \ldots, f_n\}$,
where $n$ denotes the number of faces and $f_i$ consists of $k$ integers
if the $i$-th polygon is a $k$-gon;
(3) two neighboring face indices $\langle F_1, F_2\rangle$ to start with.
The face $F_2$ must be a neighbor of the first face $F_1$.
The algorithm output is an order of face indices
$\langle F_1, F_2, \ldots, F_j\rangle$, and the unfolding is regarded as
having finished successfully if $j=n$.
We refer to the result in the following as the selection order.

\begin{enumerate}
\item Allocate a target polyhedron in the 3D Cartesian coordinate space
  such that the centroid of the polyhedron is at the origin of the
  coordinates.
\item Set the counting variable $k=2$, and rotate the polyhedron to set
  the centroid of the first polygon $F_1$ on the top; this means the
  rotation of the polyhedron sets the position vector of the centroid of
  $F_1$ on the $z$-axis.
\item Select a face $F_{k+1}$ among the adjacent faces of $F_k$ based on
  the following priority:
  \begin{enumerate}
  \item Select the remaining face automatically as $F_{k+1}$ if it is
    the only remaining neighbor.
  \item First, divide the Cartesian space into two parts by the plane
    composed of the $z$-axis and the centroid of $F_k$.
    Second, list the adjacent polygons whose centroids belong to the left
    half-space.
    Finally, select the face which has the largest $z$-coordinate value
    as $F_{k+1}$.
  \item If there are no neighboring faces in the left half-space,
    select the neighbor with the lowest $z$-coordinate.
  \end{enumerate}
\item Go back to step~3 unless all the faces have been selected or if the
  next face has no neighboring faces.
\end{enumerate}

\begin{algorithm}[tb]
\caption{Algorithm of apple peeling a polyhedron}
\label{peeling}
\begin{algorithmic}[1]
\Require $F_2$ must be one of the neighbors of $F_1$.
\Function{apple\_peeling}{$V=\{v_1, v_2, \ldots, v_m\}$,
  $F=\{f_1, f_2, \ldots, f_n\}$, $I=\langle F_1, F_2\rangle$}
  \State $k \gets 2$
  \State $l \gets$ a list of unselected neighbors of $F_{k}$
  \State $i \gets$ number of elements of $l$
  \State $c \gets$ list of centroids calculated with $V$ and $F$
  \Repeat
    \If{$i = 1$}
      \State $F_{k+1} \gets$ the face number listed in $l$
    \ElsIf{$i > 1$}
      \State $q \gets$ number of centroids in the left half-space
      \State $s \gets$ list of centroids in the left half-space
      \If{$q > 0$}
        \State $F_{k+1} \gets$ the face whose centroid $z$-coordinate
               is highest among $s$
      \Else
        \State $F_{k+1} \gets$ the face whose centroid $z$-coordinate
               is lowest among all remaining neighbors
      \EndIf
    \EndIf
    \State $k \gets k + 1$
    \State $l \gets$ a list of unselected neighbors of $F_{k}$
    \State $i \gets$ number of elements of $l$
  \Until{$i = 0$}
  \State \Return $\langle F_1, F_2, \ldots, F_{k}\rangle$
\EndFunction
\end{algorithmic}
\end{algorithm}

For the apple peeling algorithm, we guarantee that the procedure of the
algorithm yields a unique sequence by the following lemma.

\begin{lemma}
Let $\mathcal{P}$ be a polyhedron.
The sequence of faces $\langle f_1, f_2, \ldots, f_n\rangle$ obtained by
the apple peeling algorithm is unique up to the choice of the initial two
faces.
\end{lemma}

\begin{proof}
We use mathematical induction to prove the argument by starting with $n=3$.
For a given polyhedron $\mathcal{P}$ with $n$ faces; we assume that faces
$f_1, \ldots, f_n$ have centroids $c_1, \ldots, c_n$, respectively.
Let $f_{\alpha}, f_{\beta}$ be the first two faces from the set $F$,
i.e., $f_{\alpha}=F_1$ and $f_{\beta}=F_2$, in the sequence generated by
Algorithm~\ref{peeling}.

We remark that for a given polyhedron $\mathcal{P}$, the centroids of
neighboring faces of the face $f_k$ on the left-hand side have different
$z$-coordinates.
This is because the different orientations of faces with respect to the
edges of the face $f_k$ affect the positions of the faces' centroids.

For the basis step, we choose the third face $F_3$ of the sequence
satisfying step~3.
Let $f_{\ell,1}, \ldots, f_{\ell,j}$ be the neighboring faces of face
$F_2$.
Then, there exists a plane passing through $\hat{x}_3$ and $\hat{F}_2$,
which contains the unique face whose centroid is on the left-hand side of
the plane $P_{1,2}$ and is the highest among the neighboring faces.
Then, the third face can be chosen uniquely.

To prove the inductive step, we assume that the sequence of faces is
chosen until the $k$-th face, i.e., the sequence
$\langle F_1, \ldots, F_k\rangle$ has been generated.
The next choice $F_{k+1}$ can be chosen to satisfy the priority in
step~3 for all $k+1 \leq n$.
Since the determination of the left-hand side with respect to the
previous face $F_k$ is unique, and the centroids of neighboring faces
of $F_k$ are different, the choice of $F_{k+1}$ is also unique.
This concludes the proof.
\end{proof}

We remark that the peeling procedure in step~3 is specified such that
the $z$-coordinates of the face centroids are taken into consideration
when faces are chosen.
This differs from the study of~\cite{WSA2018}, which mainly focused on
spiral extraction of a polyhedral graph.

\subsection{Peelability}

In the case of peeling an apple or orange in the real world, there exists
a continuous curve on a sphere.
This curve can be spread out to be a spiral curve on a plane, as stated
in the study of~\cite{Bartholdi2012}.
Therefore, a similar concept can be considered in the case of a polyhedron.
For each polyhedron face $F_i$, we choose the centroid $c_i$ of the face
$F_i$ as a representative point.
If two faces $F_j, F_k$ are adjacent, then we assume that there is an
edge connecting the two centroids $c_j$ and $c_k$.
Therefore, we can obtain a polygonal (spiral) chain.
Since this information can be considered in a discrete structure, it is
sufficient to consider it as a sequence of faces from the apple peeling
algorithm.

The term \emph{peelable} is defined as the situation where all faces are
selected when we apply the apple peeling procedure to a polyhedron.
It is determined by the sequential selection of two neighboring faces.
By examining all two sequential selections of the neighboring faces, the
nature of a polyhedron is determined as either peelable or not.
The apple peeling procedure terminates if one of the following conditions
is satisfied: no next candidates exist or all faces are in the sequence
$\mathcal{F}$.

\begin{definition}
Let $\mathcal{P}$ be a polyhedron with a set of $n$ faces
$F=\{f_1, \ldots, f_n\}$.
Suppose that $\mathcal{F}_{\alpha}=\langle F_1, \ldots, F_t\rangle$ is
the sequence of faces generated by the apple peeling algorithm such that
$F_i = f_j$ for some $i=1, \ldots, t$ and $j=1, \ldots, n$, and
$F_1 = f_{\alpha}$ for some $\alpha$.
The polyhedron $\mathcal{P}$ is \emph{peelable} with respect to the
initial face $f_{\alpha}$ if $|\mathcal{F}_{\alpha}|=n$.
Otherwise, it is \emph{not peelable} with respect to $F_1$.
\end{definition}

\begin{definition}[Peelability]
Let $\mathcal{P}$ be a polyhedron and $\mathcal{F}_{\alpha}$ be the
sequence generated by the apple peeling procedure with the initial face
$f_{\alpha}$.
\begin{enumerate}
  \item $\mathcal{P}$ is \emph{perfectly peelable} if
        $\mathcal{F}_{\alpha}$ is a peelable sequence for all
        $\alpha=1, \ldots, n$.
  \item $\mathcal{P}$ is \emph{possibly peelable} if there exists a
        peelable sequence $\mathcal{F}_{\alpha}$, but there is a sequence
        $\mathcal{F}_{\beta}$ which is not peelable.
  \item $\mathcal{P}$ is \emph{non-peelable} if $\mathcal{F}_{\alpha}$
        is not a peelable sequence for all $\alpha=1, \ldots, n$.
\end{enumerate}
\end{definition}

\subsection{Visualization}

We used three types of visualizations in the following:
(a) a three-dimensional representation to display the sequential
selection,
(b) a polyhedron net, and
(c) a planar graph.
We have already shown an example of the three-dimensional representation
in Fig.~\ref{Fig:Example3D}.
In addition, we used red polygons to represent the remaining ones during
the peeling.
This means that the existence of red polygons indicates the failure of
peeling.
Polyhedron nets are used to display the results of the unfolding when it
finishes successfully.
We used a grayscale to represent the sequence; from dark to light gray
corresponds to from a start to a goal of peeling.
A planar graph is used to illustrate the whole trace of the selection.

The three-dimensional representation and the polyhedron net use the same
procedure: a repetition of rotations of polygons.
In the case representing the intermediate state of peeling in
three-dimensional space, we rotate the polygons $F_k$ ($k=1,2,\ldots,i$)
with respect to the edge shared by $F_i$ and $F_{i+1}$ simultaneously
such that the solid angle between the two faces is flat.
We obtain a polyhedron net by repeating rotating from $i=2$ to $i=n-1$
continuously, where $n$ denotes the number of elements of the order.
It should be noted that some disconnections were observed in the case of
some polyhedron nets for a Catalan solid due to numerical errors.
We corrected such disconnections manually as necessary.
The insides of the polyhedra are facing us for all the polyhedron nets
in the following.

The planar graph representation is functional, especially in the case of
uniform polyhedra, because the symmetry properties of the polyhedra are
mostly conserved.
First, we project the vertices from $(0, 0, \hat{x}_3)$,
$\hat{x}_3 > x_{3,1}$, to the $xy$-plane and arrange the vertices based
on the Tutte method~\cite{Tutte}.
Second, we rotate the vertices around the origin so that the projected
centroid of $F_2$ is on the $y$-axis.
Third, we draw the edges.
Finally, we add the numerical results by drawing the lines connecting the
centroids of the faces and coloring the unselected (remaining) faces red.
Consequently, the outermost polygon of the planar graph represents the
exact shape of $F_1$ and the inner structure of the graph represents the
vertex connections.

\section{Results}\label{sec:results}

\subsection{Numerical Simulations}

We carried out numerical simulations using Mathematica\texttrademark.
The primary data of polyhedra, such as the vertex coordinates and vertex
numbers of each face, were obtained from a library of Mathematica
(``PolyhedronData'' function).
Although Mathematica allows us to manipulate exact values
(e.g., $\pi$ and $\sqrt{3}$), we used approximated values to decrease
the execution time because the simulation using exact values took an
impractical amount of time.
This means that the results were numerical rather than exact.
We tried all possible combinations of $\langle F_1, F_2\rangle$ as the
initial values of the simulation and summarized the results comparisons
of their polyhedron nets.

\subsection{Platonic Solids}\label{platonic-solids}

All the Platonic solids were classified as perfect polyhedra because the
Platonic solids are not changed in nature by the selection of $F_1$ and
$F_2$, so the results are the same for all combinations of $F_1$ and
$F_2$.
Figure~\ref{Fig:Platonic} illustrates these results, which have previously
been reported~\cite{Kaino2019}.
We obtained only one type of polyhedron net for dodecahedra, although
Kaino listed two types of polyhedron nets.
This means that only the polyhedron net in Fig.~\ref{Fig:Platonic}
matches our definition of apple peel unfolding.
It is notable that we did not obtain the unfolding of a dodecahedron
labeled by Kaino as the Atake-type unfolding (``Dodecahedron~1'' in
Kaino's notation).
The net consists of ten pentagons connected in a zigzag fashion, and the
remaining two pentagons opposite each other on the top and bottom.

\begin{figure}[ht]
\centering
\resizebox*{8cm}{!}{\includegraphics{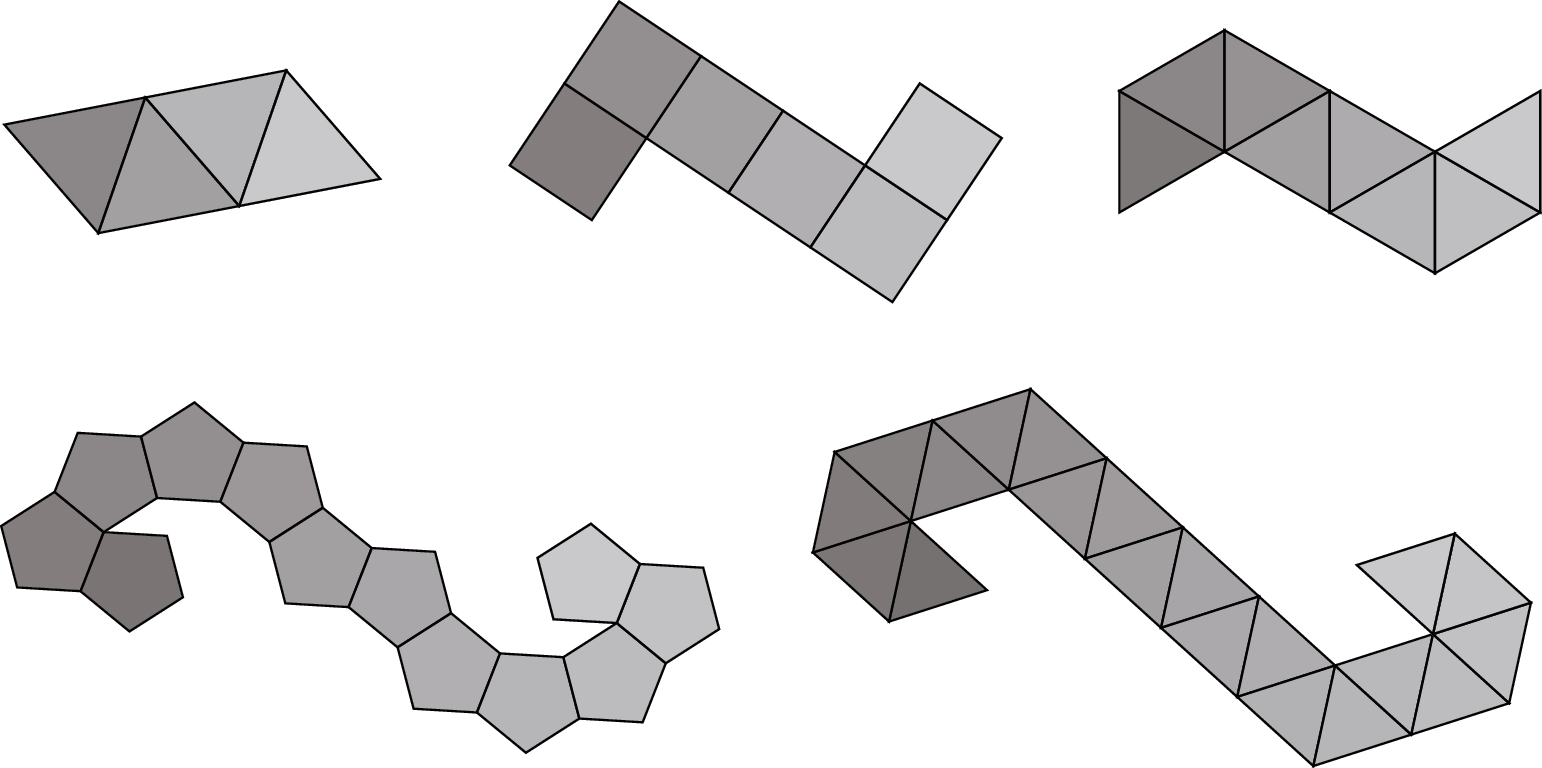}}
\caption{Polyhedron nets of Platonic solids obtained from apple peel
  unfolding.}
\label{Fig:Platonic}
\end{figure}

\subsection{Archimedean Solids}\label{archimedean-solids}

The results for Archimedean and Catalan solids are summarized in
Fig.~\ref{fig:AllNets} and Table~\ref{tab:summary}.
Figure~\ref{fig:AllNets} illustrates the examples of the polyhedron nets
for which the unfolding has finished successfully for all peelable
polyhedra.
For Archimedean solids, the numbers of perfect, possible, and impossible
polyhedra are three, three, and seven, respectively.
On the other hand, those for Catalan solids are six, three, and four,
respectively.
There is no apparent relation between the results of Archimedean solids
and those of their duals.
Therefore, we conclude that the symmetry properties of the faces and
vertices do not play a critical role in peelability.

We included information on the existence of the Hamiltonian paths in
Table~\ref{tab:summary} because apple peeling is the problem of finding
Hamiltonian paths under restricted conditions.
Seven Catalan solids have Hamiltonian paths, as shown in the table.
This suggests that some Archimedean solids intrinsically cannot be
apple-peeled.
We omit the information on the Hamiltonian paths of the Archimedean solid
in the table because all solids have Hamiltonian paths.
From this viewpoint, the ratios of the number of peelable solids to the
number of solids having Hamiltonian paths are 6/7 and 9/13 for
Archimedean and Catalan solids, respectively.

\begin{figure}[ht]
\centering
\resizebox*{14cm}{!}{\includegraphics{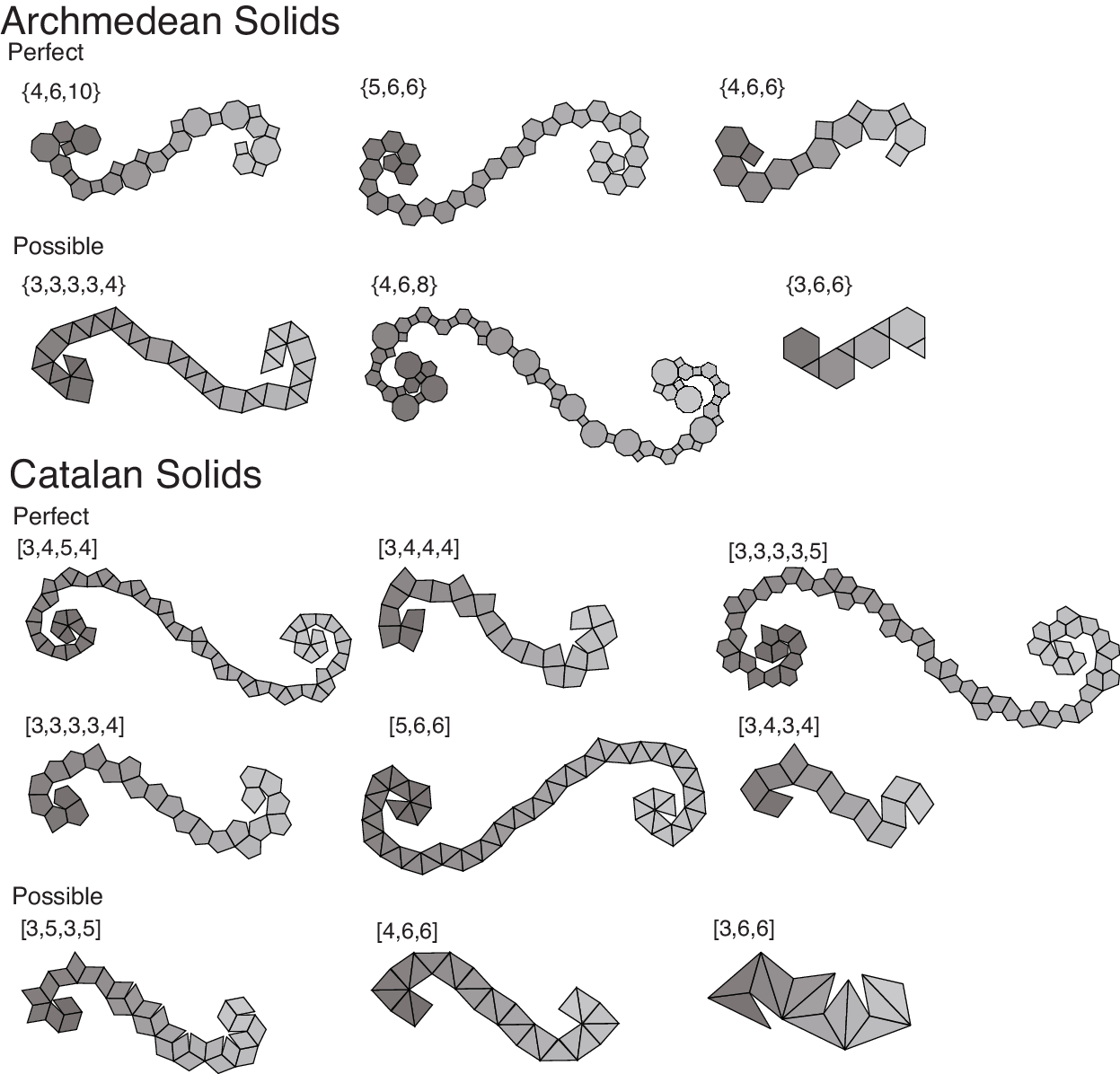}}
\caption{Examples of polyhedron nets of all polyhedra for which they were
  successfully completed.}
\label{fig:AllNets}
\end{figure}

\begin{table}[ht]
\centering
\caption{Summary of numerical simulations.
  Bold text indicates that complete results are presented herein.}
\label{tab:summary}
\begin{tabular}{llccc}
\toprule
Name & Index
  & \multicolumn{2}{c}{Peelability}
  & Hamiltonian path \\
  & & Archimedean & Dual (Catalan) & of Catalan solids \\
\midrule
Cuboctahedron                 & $\{3,4,3,4\}$   & Impossible & Perfect     & n.a. \\
Great Rhombicosidodecahedron  & $\{4,6,10\}$    & Possible   & Impossible  & available \\
Great Rhombicuboctahedron     & $\{4,6,8\}$     & Perfect    & Impossible  & available \\
Icosidodecahedron             & $\{3,5,3,5\}$   & Impossible & Possible    & n.a. \\
Small Rhombicosidodecahedron  & $\{3,4,5,4\}$   & Impossible & Perfect     & n.a. \\
Deltoidal Icositetrahedron    & $\{3,4,4,4\}$   & Impossible & Perfect     & n.a. \\
Snub Cube                     & $\{3,3,3,3,4\}$ & \textbf{Possible} & Perfect & available \\
Snub Dodecahedron             & $\{3,3,3,3,5\}$ & Impossible & Perfect     & available \\
Truncated Cube                & $\{3,8,8\}$     & Impossible & \textbf{Impossible} & n.a. \\
Truncated Dodecahedron        & $\{3,10,10\}$   & \textbf{Impossible} & Impossible & n.a. \\
Truncated Icosahedron         & $\{5,6,6\}$     & \textbf{Perfect} & \textbf{Perfect} & available \\
Truncated Octahedron          & $\{4,6,6\}$     & Perfect    & \textbf{Possible} & available \\
Truncated Tetrahedron         & $\{3,6,6\}$     & Possible   & Possible    & available \\
\bottomrule
\end{tabular}
\end{table}

All perfect Archimedean solids contain hexagons.
The polyhedron nets of a truncated icosahedron $\{5,6,6\}$ are summarized
as the three types shown in Fig.~\ref{fig:11Nets} as examples of perfect
solids.
What differs among the three types is the appearance of the first
pentagon; they are $F_1=5$, $F_2=5$, and $F_3=5$ from left to right.
The results of the polyhedron were the same as the spiral naming of
$\mathrm{C}_{60}$ fullerene~\cite{Manolopoulos1991}.

\begin{figure}[ht]
\centering
\resizebox*{12cm}{!}{\includegraphics{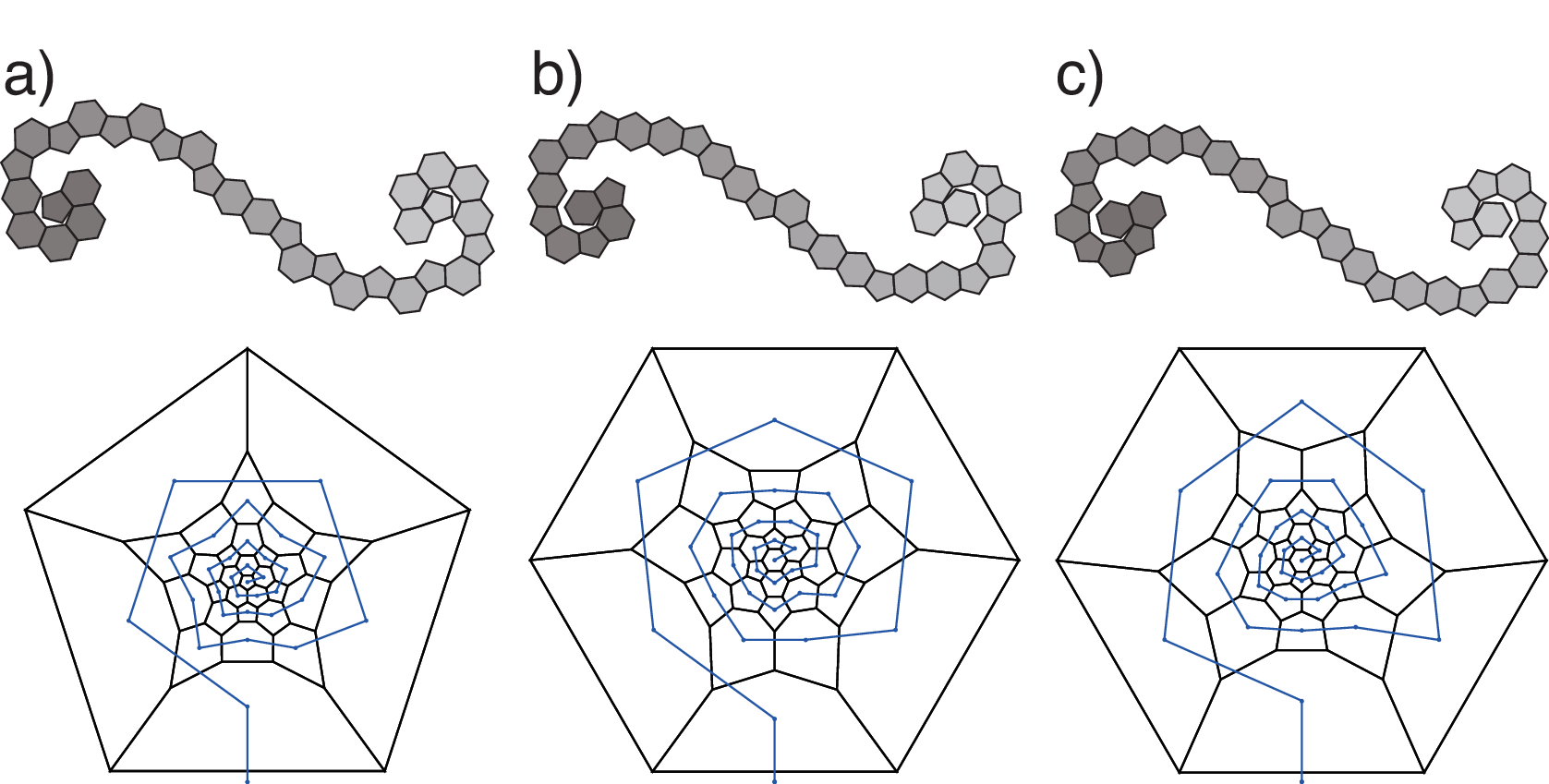}}
\caption{Results of apple peel unfolding of truncated icosahedron
  $\{5,6,6\}$: (a) $F_1=5$, (b) $F_2=5$, and (c) $F_3=5$.}
\label{fig:11Nets}
\end{figure}

The results for a snub cube are summarized in Fig.~\ref{fig:07NetsAndGraphs}
with polyhedron nets and planar graphs.
The three patterns in the top row
(Figs.~\ref{fig:07NetsAndGraphs}a,~b, and~c) represent the successful
cases, and the rest are the failed cases.
The peelings from a square have two patterns
(Figs.~\ref{fig:07NetsAndGraphs}a and~d), and those from a triangle have
four.
Chirality affects the results in this case.
Comparison of Figs.~\ref{fig:07NetsAndGraphs}a and~d shows the influence
in the cases starting from a square.
The appearance of the second square is one of the differences between the
two patterns.
The peeling succeeded if the second square appeared as the fourteenth
polygon and failed if the square appeared as the sixteenth.
In the case of a left-handed person, unfolding succeeded in
Fig.~\ref{fig:07NetsAndGraphs}d and failed in
Fig.~\ref{fig:07NetsAndGraphs}a.
There are thirteen triangles between the first and second squares for the
peelable cases.
On the other hand, the effect appeared in different ways in the cases
starting from a regular triangle.
The results were classified into four types, illustrated in
Figs.~\ref{fig:07NetsAndGraphs}b,~c,~e, and~f.

The planar graph of Fig.~\ref{fig:07NetsAndGraphs}b has three-fold
rotational symmetry.
Therefore, any choice as $F_2$ among A, B, or C in
Fig.~\ref{fig:07NetsAndGraphs}b gives the same result.
Furthermore, it was peelable for left-handed people.
On the other hand, the choice among A, B, or C in
Fig.~\ref{fig:07NetsAndGraphs}c affected the result.
Only the case of selecting A as $F_2$ in Fig.~\ref{fig:07NetsAndGraphs}c
was peelable, and selecting B and C corresponded to
Fig.~\ref{fig:07NetsAndGraphs}e and Fig.~\ref{fig:07NetsAndGraphs}f,
respectively.

It should be noted that a snub dodecahedron, another Archimedean solid
having chirality, was an impossible polyhedron.
In other words, all the trials of its apple peeling unfolding failed.

\begin{figure}[ht]
\centering
\resizebox*{12cm}{!}{\includegraphics{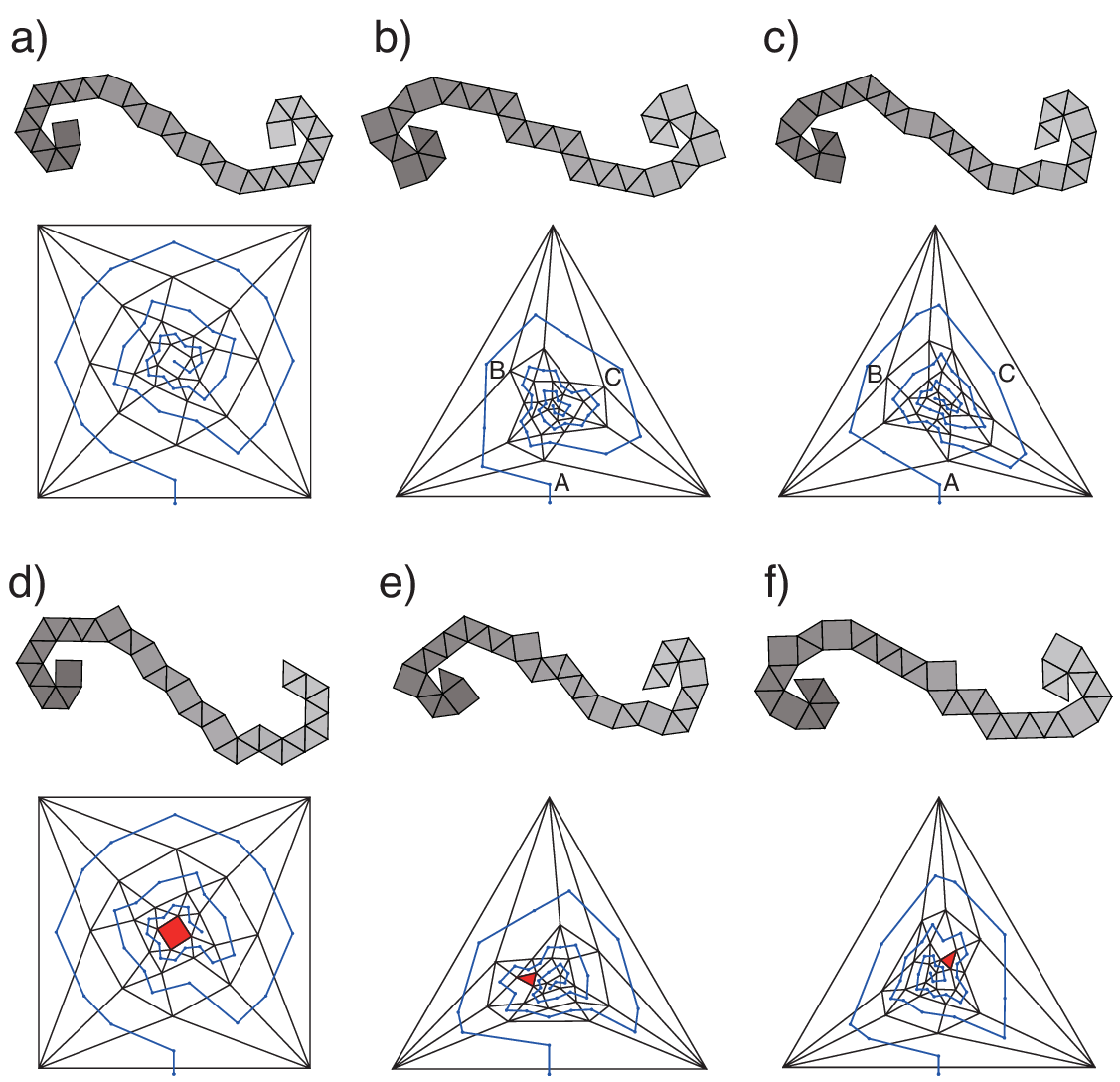}}
\caption{Comparison between a polyhedron net and a planar graph of a snub
  cube $\{3,3,3,3,4\}$.}
\label{fig:07NetsAndGraphs}
\end{figure}

Figure~\ref{fig:TruncatedDodecahedron} shows the results of truncated
dodecahedron $\{3,10,10\}$: an example of an impossible polyhedron.
We found some variations originated from numerical errors.
Although we obtained seven patterns for truncated dodecahedron
$\{3,10,10\}$, shown in Fig.~\ref{fig:TruncatedDodecahedron}, the basic
patterns were only the first three
(Figs.~\ref{fig:TruncatedDodecahedron}a,~b, and~c).
The other four results have the same structure as
Fig.~\ref{fig:TruncatedDodecahedron}c.

\begin{figure}[ht]
\centering
\resizebox*{12cm}{!}{\includegraphics{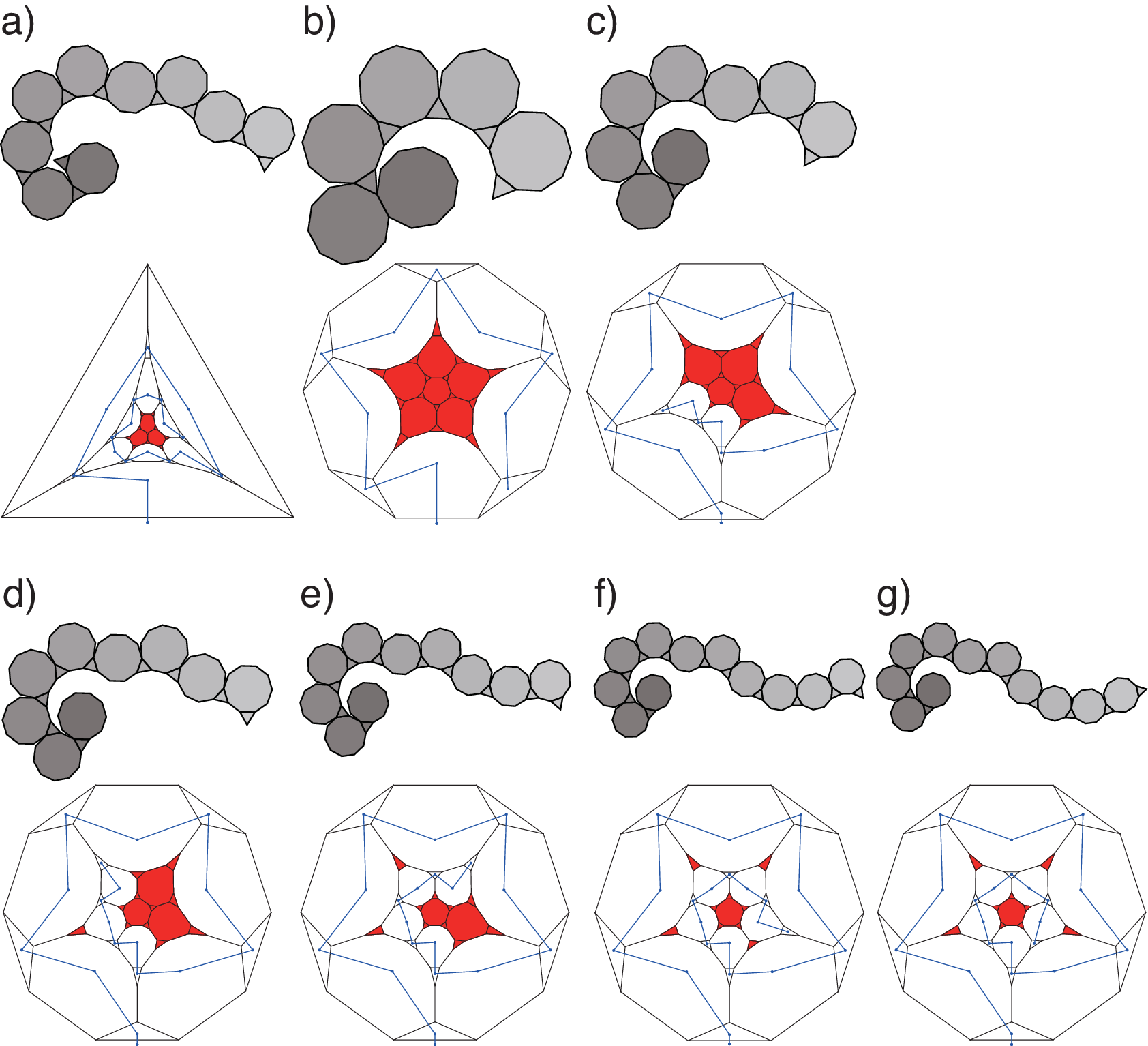}}
\caption{Summary of peelings of truncated dodecahedron $\{3,10,10\}$.}
\label{fig:TruncatedDodecahedron}
\end{figure}

\subsection{Catalan Solids}\label{catalan-solids}

\begin{figure}[ht]
\centering
\resizebox*{14cm}{!}{\includegraphics{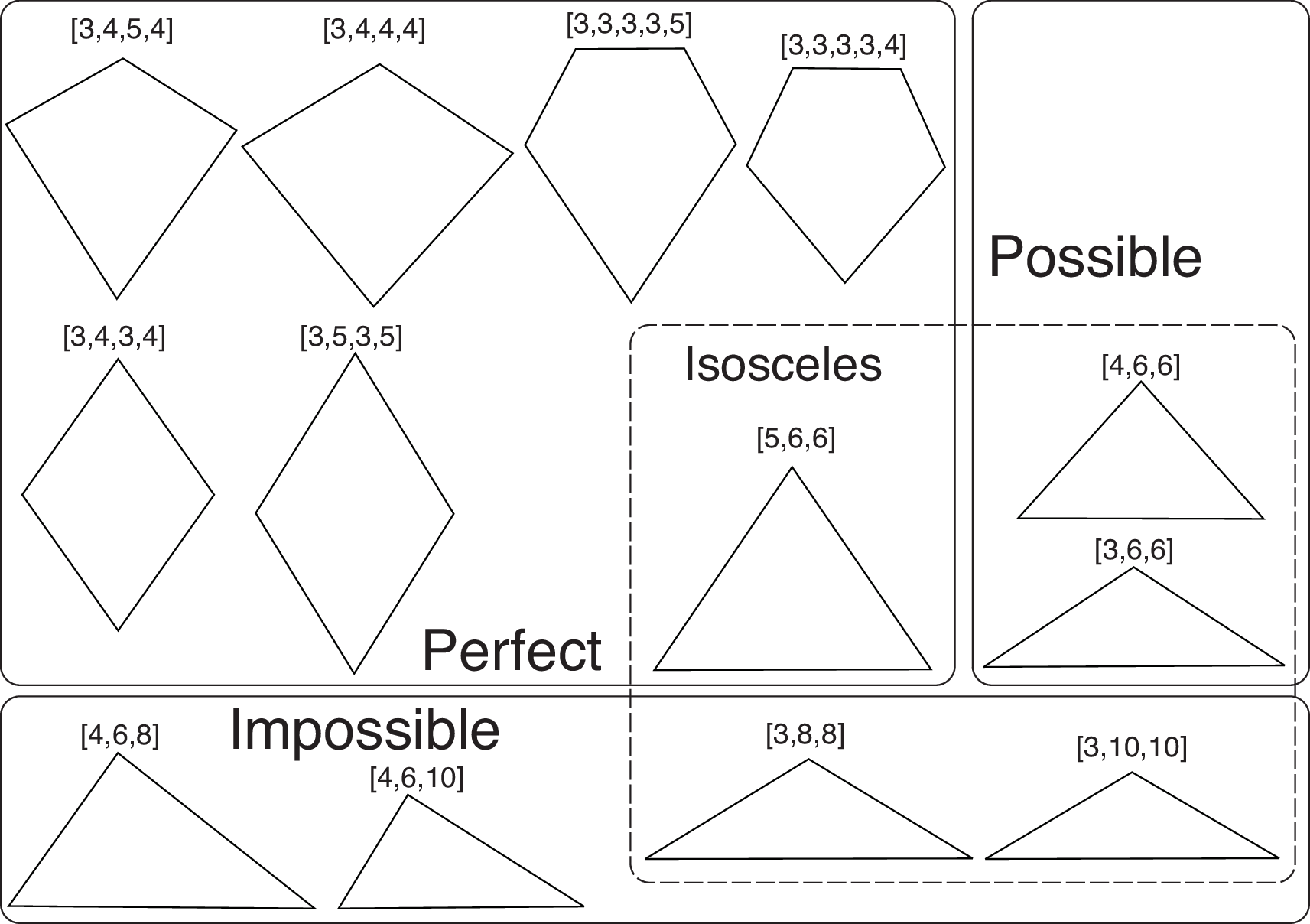}}
\caption{Facets of Catalan solids and their peelability.}
\label{Fig:CatalanFacets}
\end{figure}

Nine polyhedra were successfully peeled by apple peel unfolding as
summarized in Table~\ref{tab:summary} and Fig.~\ref{fig:AllNets}.
In addition, a comparison between the results and the shapes of the facets
is presented in Fig.~\ref{Fig:CatalanFacets}.
The polyhedra having pentagons ($[3,3,3,3,4]$ and $[3,3,3,3,5]$),
kites ($[3,4,4,4]$, $[3,4,5,4]$), and rhombuses ($[3,4,3,4]$,
$[3,5,3,5]$) were shown to be perfect.
The results of three Catalan solids are presented as examples in
Figs.~\ref{Fig:7D},~\ref{Fig:11D}, and~\ref{Fig:03D}.

All three examples are polyhedra consisting of isosceles triangles.
The pentakis dodecahedron $[5,6,6]$ was classified as perfect: all
peelings succeeded.
Figure~\ref{Fig:7D} demonstrates three types of unfolding, and the
difference arose from the connection of the first two faces.
The result of the tetrakis hexahedron $[4,6,6]$, as an example of a
possible polyhedron, is shown in Fig.~\ref{Fig:11D}.
The quadrilateral is the first polyhedron, which differs among the nets.
All trials of small triakis octahedron $[3,8,8]$ were classified into
three nets, as shown in Fig.~\ref{Fig:03D}.

\begin{figure}[ht]
\centering
\resizebox*{14cm}{!}{\includegraphics{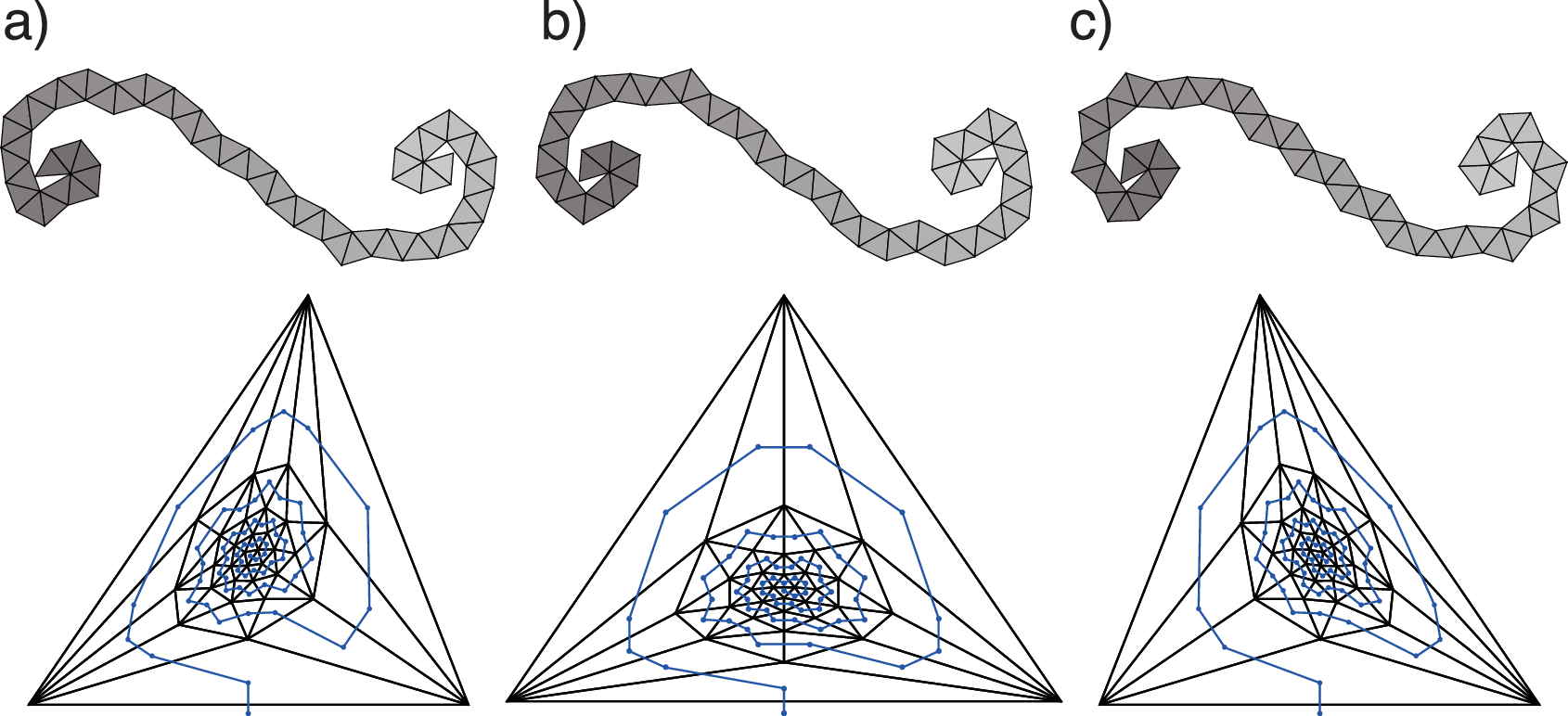}}
\caption{Summary of cases for pentakis dodecahedron $[5,6,6]$.}
\label{Fig:7D}
\end{figure}

\begin{figure}[ht]
\centering
\resizebox*{14cm}{!}{\includegraphics{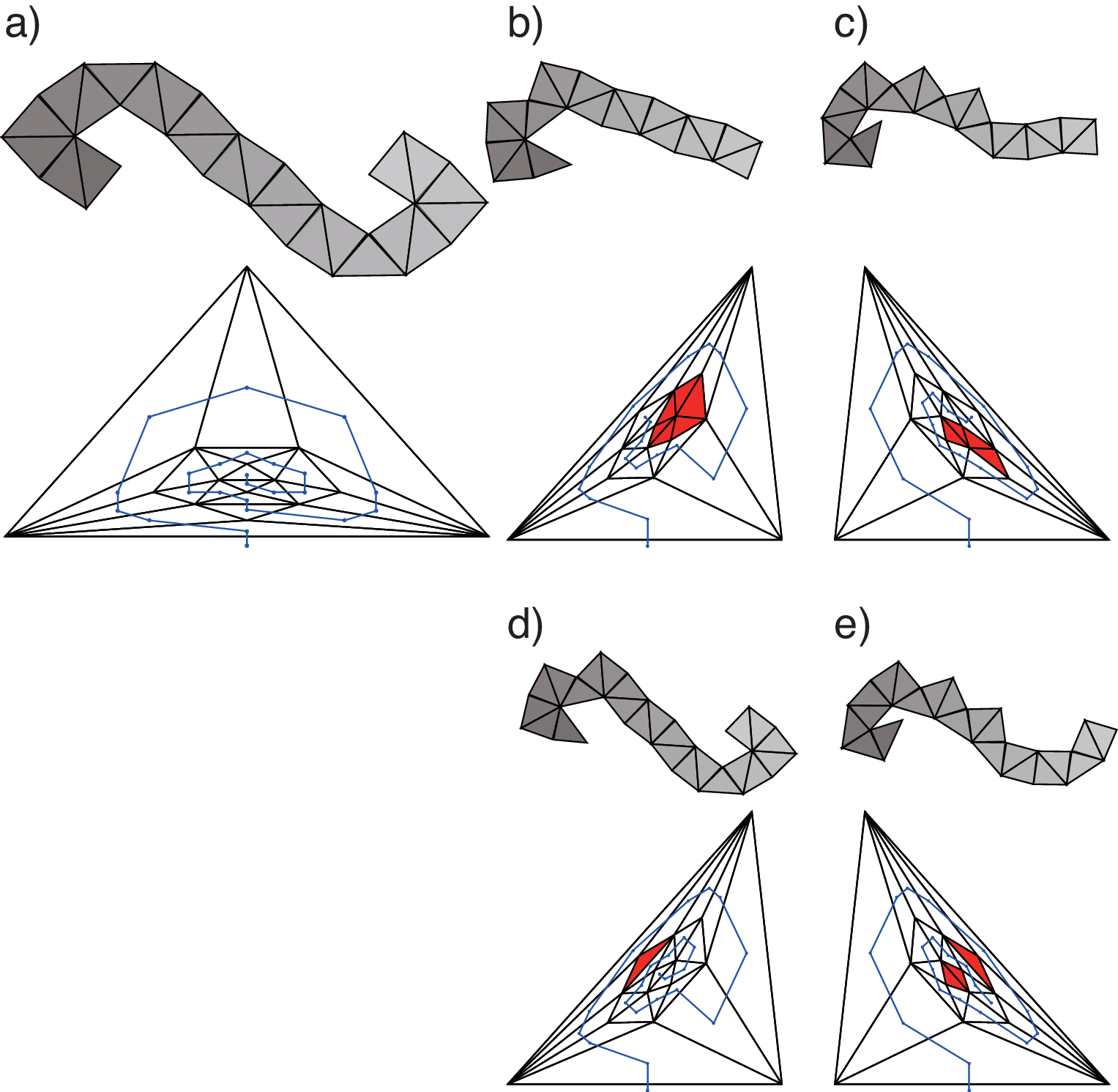}}
\caption{Summary of cases for tetrakis hexahedron $[4,6,6]$.}
\label{Fig:11D}
\end{figure}

\begin{figure}[ht]
\centering
\resizebox*{14cm}{!}{\includegraphics{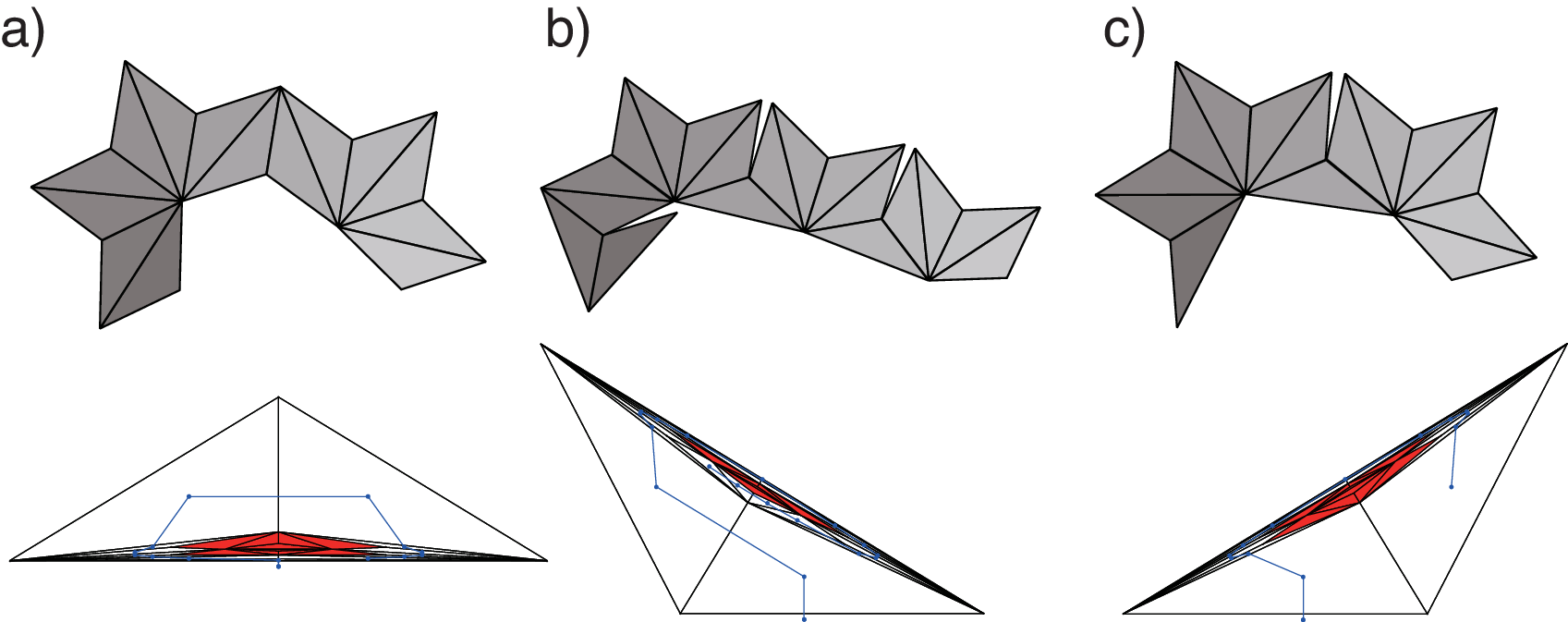}}
\caption{Summary of cases for small triakis octahedron $[3,8,8]$.}
\label{Fig:03D}
\end{figure}

\section{Discussion}\label{sec:discussions}

Using the proposed procedure, we successfully obtained results of apple
peel unfolding for Archimedean and Catalan solids.
The proposed definition and the algorithm to obtain the results seem
appropriate; however, some issues remain to be discussed.
These are extensions to general cases, the criterion of peelability, and
the possibility of other definitions.

Regarding whether our definition of apple peel unfolding can be extended
to general cases, one extension is to apply the definition to random
polyhedra.
The present definition seems to be appropriate for random polyhedra.
Another extension is to higher-dimensional polytopes.
We expect that the definition is applicable to four-dimensional polytopes
because it is consistent with Kaino's apple peeling of four-dimensional
polytopes.

It has not been clarified yet what geometrical property determines the
peelability of a given solid.
Although we tried to find such a property by considering dualities and
Hamiltonian paths, we were not able to.
The results for the Catalan solids having triangular faces indicate that a
flat shape tends to be not peelable; however, it is not clear how flatness
relates to peelability even qualitatively.

Failures of peelings can be classified into two types, although the
termination condition is only ``no candidate''.
One is the case that all neighbors have already been selected
(termination), and the other is that unselected face(s) were skipped
(isolation).
Although we could not classify them strictly, the difference will be one
of the keys to understanding the results.

Although our definition of apple peel unfolding provided fruitful results,
we should discuss the validity precisely because there is another unfolding
of a dodecahedron: the Atake-type unfolding.
This unfolding could not be obtained in the numerical simulation; however,
it can be said that this type is appropriate as another kind of apple
peeling.
The problem of extending our definition of unfolding to match the
Atake-type unfolding remains.

\section*{Acknowledgements}

One of the authors (T.Y.) thanks Prof.\ Keimei Kaino, emeritus professor
of the National Institute of Technology Sendai College, for comments on
the results during the preparation of the manuscript.

\section*{Disclosure statement}

The authors declare that they have no conflicts of interest concerning this
article.

\section*{Funding}

This research was supported by the Toyo University Short-term
International Visiting Professor Program 2022, and partially supported by
Chiang Mai University.



\begin{thebibliography}{9}

\bibitem{Demaine2010}
E.~D.~Demaine, M.~L.~Demaine, A.~Lubiw, A.~Shallit, and J.~L.~Shallit,
Zipper unfoldings of polyhedral complexes,
in \textit{Proc.\ 22nd Canadian Conf.\ Comput.\ Geom.\ (CCCG 2010)},
Winnipeg, August 2010.

\bibitem{ORourke2015}
J.~O'Rourke,
Spiral unfoldings of convex polyhedra,
\textit{arXiv preprint arXiv:1509.00321}, 2015.

\bibitem{Kaino2019}
K.~Kaino,
Apple-peel foldouts of four-dimensional regular polytopes: 24, 120 and
600-cells,
\textit{Symmetry: Art and Science}, pp.~25--30, 2019.

\bibitem{Tutte}
W.~T.~Tutte,
How to draw a graph,
\textit{Proc.\ London Math.\ Soc.}, \textbf{3}(1) (1963), 743--767.

\bibitem{Manolopoulos1991}
D.~E.~Manolopoulos, J.~C.~May, and S.~E.~Down,
Theoretical studies of the fullerenes: $\mathrm{C}_{34}$ to
$\mathrm{C}_{70}$,
\textit{Chemical Physics Letters}, \textbf{181}(2--3) (1991), 105--111.

\bibitem{Manolopoulos1993}
D.~E.~Manolopoulos and P.~W.~Fowler,
A fullerene without a spiral,
\textit{Chemical Physics Letters}, \textbf{204}(1--2) (1993), 1--7.

\bibitem{WSA2018}
L.~N.~Wirz, P.~Schwerdtfeger, and J.~E.~Avery,
Naming polyhedra by general face-spirals---theory and applications to
fullerenes and other polyhedral molecules,
\textit{Fullerenes, Nanotubes and Carbon Nanostructures},
\textbf{26}(10) (2018), 607--630.

\bibitem{Bartholdi2012}
L.~Bartholdi and A.~Henriques,
Orange peels and Fresnel integrals,
\textit{Mathematical Intelligencer}, \textbf{34} (2012), 1--3.

\bibitem{Uehara2020}
R.~Uehara,
\textit{Introduction to Computational Origami},
Springer Singapore, 2020.

\end{thebibliography}
\end{document}